\documentclass[aps,twocolumn,showpacs,pra,nofootinbib]{revtex4}
\usepackage{bm}
\usepackage{graphicx}
\usepackage{amsbsy}
\usepackage{amsmath}
\usepackage{amsfonts}
\usepackage{amssymb}
\usepackage{amsmath}
\usepackage{amsthm}
\usepackage{hyperref}
\usepackage{url,doi}
\usepackage{url}
\usepackage{breakurl}

\newtheorem*{lemA.1}{Lemma A.1}
\newtheorem*{lemA.2}{Lemma A.2}
\newtheorem*{lemA.3}{Lemma A.3}
\newtheorem{theorem}{Theorem}

\newtheorem{lemma}{Lemma}
\newtheorem{definition}{Definition}
\newtheorem{claim}{Claim}
\newtheorem{problem}{Problem}

\newcommand{\color}[3]{}

\newcommand{\op}[1]{\mathsf{#1}}
\newcommand{\nth}[1]{$#1^{\textup{th}}$}

\begin{document}
\title{Quantum Computation with Coherent Spin States and the Close Hadamard Problem}
    \author{Mark R. A. Adcock,$^{1}$ Peter H\o yer,$^{1,2}$ and Barry C. Sanders$^{1,3}$
   }
 \affiliation{$^{1}$Institute for Quantum Science and Technology,
    University of Calgary, Calgary, Alberta, Canada, T2N 1N4. Email: \texttt{mkadcock@qis.ucalgary.ca}\\
 $^{2}$Department of Computer Science,University of Calgary, 2500 University Drive N.W., Calgary,
Alberta, Canada, T2N 1N4. Email: \texttt{hoyer@ucalgary.ca}\\
 $^{3}$Program in Quantum Information Science, Canadian Institute for Advanced Research, Toronto, Ontario M5G 1Z8,
 Canada
 }
\begin{abstract}

We study a model of quantum computation based on the
continuously-parameterized yet finite-dimensional Hilbert space of
a spin system. We explore the computational powers of this model
by analyzing a pilot problem we refer to as the close Hadamard
problem. We prove that the close Hadamard problem can be solved in
the spin system model with arbitrarily small error probability in
a constant number of oracle queries. We conclude that this model
of quantum computation is suitable for solving certain types of
problems. The model is effective for problems where symmetries
between the structure of the information associated with the
problem and the structure of the unitary operators employed in the
quantum algorithm can be exploited.

\end{abstract}
 \maketitle

\section{Introduction}

Ever since the remarkable discovery that quantum mechanical
systems can in principle be used for computational purposes by
pioneers such as Benioff \cite{Be82}, Feynman \cite{Fe82}, and
Deutsch \cite{De85}, researchers have considered the feasibilities
of implementing such computations in a rich variety of physical
models. These models include phenomena such as light, electrons,
atomic nuclei each of which have degrees of freedom such as
position, path or spin into which quantum information can be
encoded, processed and measured. The measurement spectrum can be
discrete, for example measuring spin up or down or position being
left or right, or continuous, for example measuring the angle of
the spin axis or where on a line the particle is located.

This distinction between whether the spectrum of measurement, or
preparation, is discrete or continuous is at the heart of the
difference between discrete variable and continuous variable
versions of quantum information~\cite{Ad12}. Continuous variable
studies are most often linked to harmonic oscillators because
quantum optics has powerful tools to prepare, process and measure
optical field modes~\cite{Le97}, which are analogous to harmonic
oscillators.

The harmonic oscillator is a continuously-parameterized system
having an infinite-dimensional Hilbert space. Continuous variable
quantum computation based on the harmonic oscillator may be
thought of as a model of quantum computation. The circuit model
was the first such model~\cite{De89}, and the circuit and one-way
quantum models are examples of different models~\cite{Ne06}, which
are distinguished by different processing methods. Since
continuous-variable quantum information may be defined in both
finite- and infinite- dimensional Hilbert spaces~\cite{Ad12}, we
are inspired to investigate computational models in continuously
parameterized yet finite-dimensional Hilbert spaces.

Here we study a model of quantum computation based on the
continuously-parameterized, finite-dimensional Hilbert space of a
spin system~\cite{Pe72,Ar72,Ra71,KU93}. We are ultimately
interested in characterizing the computational powers of this
system. We initiate this investigation by proposing and analyzing
a pilot problem we refer to as the close Hadamard problem. This
problem has in particular two advantages: Firstly, this problem
seems suitable for solution within a finite-dimensional model.
Secondly, since variations of this problem have been studied in
other models, it provides a mean to compare this model with other
models of quantum computing.

The close Hadamard problem oracle is an oracle decision problem
~\cite{AHS09,AHS12}, related to the digital coding techniques
employed in classical communications~\cite{MS77,Hor07}. It is also
a special case of what is sometimes referred to as the
bounded-distance decoding problem~\cite{BV93,MM99}. In order to
highlight features of the spin model, we differentiate between two
versions of the close Hadamard problem, which we refer to as the
restricted and the unrestricted versions.

We introduce a new algorithm and prove that both versions of the
problem can be solved in the spin system model with arbitrarily
small error probability in a constant number of oracle queries.
This performance is comparable to that of the Bernstein--Vazirani
algorithm~\cite{BV93,MM99} adapted to this problem in the circuit
model of quantum computation. The observation of comparable
performance does not imply a win over the Bernstein--Vazirani
algorithm in the circuit model. That the performance is comparable
in this new model of computation in and of itself warrants further
exploration of the model.

Our investigation of the pilot problem suggests that this model
can be used to exploit symmetries between the structure of the
information associated with the problem and the structure of the
unitary operators employed in the quantum algorithm. The improved
efficiency in the restricted case over the unrestricted case of
the pilot problem is due to error cancellation that results from
employing a symmetric superposition of spin states as algorithm
input and from the combination of the group structure of the
Hadamard codewords and the employment of Hadamard operators in the
algorithm.

The tolerance of errors in this case is a direct result of error
cancellation that results from this combination. We further
demonstrate that this pairing between operators and the codewords
offers the promise of discovery of new problems by giving a sketch
of another problem that can be efficiently solved in this way. We
conclude that the continuously-parameterized representation of
quantum dynamical systems having a finite-dimensional Hilbert
space gives us a model of quantum computation that is suitable for
certain types of problems and is worthy of further exploration.

Our paper is presented as follows. In Sec.\ II, we give formal
definitions of oracle decision problems and of the unrestricted
and the restricted versions of the close Hadamard problem. In
Sec.\ III, we introduce the spin system model and spin
squeezing~\cite{KU93}. We show that for a particular coherent spin
state, the limiting squeezed state is asymptotically approximated
by a symmetric superposition of two discrete states with constant
error independent of the size of the Hilbert space.

In Sec.\ IV, we prove that our algorithm solves the close Hadamard
problem with arbitrarily small error probability in a constant
number of oracle queries independent of the size of the problem.
In Sec.\ V, we discuss generalization of the computational model
by showing that if the Hadamard operator is replaced by the
discrete Fourier transformation, the oracle decision problem
changes. We conclude in Sec.\ VI.

\section{Oracle Decision Problems and the Close Hadamard Problem}

Quantum algorithms for the efficient solution of oracle decision
problems are of historical importance in quantum information
~\cite{De85,DJ92}. The Deutsch--Jozsa problem in particular has
been studied in both discrete and continuous variable settings
~\cite{DJ92,BP03,AHS09,AHS12}. We are inspired by the continuous
variable quantum algorithm used to solve the Deutsch--Jozsa
problem, where logical states are a continuously-parameterized in
a \emph{infinite}-dimensional Hilbert space. Here we explore this
continuous variable quantum algorithm where the logical states are
continuously-parameterized in the \emph{finite}-dimensional
Hilbert space of a spin system.

This continuously-parameterized spin model of quantum computation
naturally yields to a symmetric superposition of two basis states
as the logical state. Returning to a discrete representation
inspires us to discover a new oracle decision problem, which we
refer to as the close Hadamard problem.

\subsection{Oracle Decision Problems}

Oracle decision problems are related to oracle identification
problems, which are usually presented in terms of the problem of
identifying a unique function $f$. The function $f$ maps
$N=2^n$-bit strings to a single bit
\begin{align}\label{classical_gl1}
f:\{0,1\}^n\mapsto\{0,1\}.
\end{align} Any Boolean function on $n$ bits can also be represented by a
string of $N=2^n$ bits, in which the $i^{\rm{th}}$ bit $z _i$ is
the value of the function on the $i^{\rm{th}}$ bit string, taken
in lexicographical order. The challenge of the oracle
identification problem is to identify a unique $N$-bit string from
a set of size $2^N$ by making the fewest queries to an oracle.

Oracle decision problems are simpler than oracle identification
problems because the function or string does not have to be
identified explicitly. Rather, the problem is to identify which of
two mutually disjoint sets contains the string. For our analysis,
the oracle decision problem is defined as follows.
\begin{definition}\label{def:1}
An oracle decision problem is specified by two non-empty, disjoint
subsets $A,B \subset \{0,1\}^N$. Given a string $z \in A\cup B=C$,
the oracle-decision problem is to determine whether $z\in A$ or
$z\in B$ with the fewest queries to the oracle possible.
\end{definition}

\subsection{The Close Hadamard Problem}

In this subsection, we specify the particular sets $A$,~$B$
and~$C$ required by Definition 1 for the close Hadamard problem.
We are interested in strings that are close in the sense of
Hamming distance to the $N=2^n$-bit strings referred to as
Hadamard codewords~\cite{Hal69,MS77}.

The problem of discriminating between codewords received after
transmission over a noisy channel is well-known in classical
digital coding theory employing linear block codes~\cite{ML85}.
Linear block codes are characterized by the triplet $[N,k,t]$,
where $N$ is the total length of the codeword; $k<N$ is the amount
of information coded, and $t-1$ is the number of errors that the
code can correct.

The Hadamard code is a linear block code with $N=2^n$, $k=n+1$,
and $t=N/4-1$. The Hadamard code has a poor information rate
$k/N$, but it has excellent error-correcting capability. Because
of this latter feature, the $[32,6,7]$ Hadamard code was used to
encode picture information on Mariner space craft missions
\cite{Hal69}.

For $N=2^n$ and $i,j\in \{0,1\}^n$, the matrix of Hadamard
codewords is defined as
\begin{align}
\op{W}^{(N)}=\left[w_{ij}=i \cdot j\right], \label{Sec1EqnLogH}
\end{align}
where $i \cdot j$ is the$\mod{2}$ bit-wise dot product between the
the matrix indices  $i$ and $j$. For the specific $N =4$ case
where $i=j=11$, $i \cdot j=1+1=0 \mod 2$. Similarly, the other 15
bit-wise dot products give
\begin{align}
\op{W}^{(4)}=\left(\begin{array}{cccc}
0 &   0    &      0     & 0\\
0 &   1    &      0     & 1 \\
0 &   0    &      1     & 1\\
0 &   1    &      1     & 0\\
\end{array} \right).\label{eqn:HadamardCodeW4}
\end{align} For $j\in \mathbb{Z}_N$, the $j^{\rm th}$ Hadamard codeword
corresponds to the $j^{\rm th}$ row of the matrix $\op{W}^{(N)}$
and is expressed as $W^{(N)}_j$. For example $W^{(4)}_3=0110$.

All Hadamard codewords are balanced with the exception of
$W_0^{(N)}$, which is constant. Additionally, all $N$ Hadamard
codewords are separated from each other by Hamming distance
\begin{align}d\left(W^{(N)}_j,W^{(N)}_k\right)=N/2.
\end{align} An arbitrary string $z\in {\{0,1\}}^N$ having Hamming distance
$d\left(z,W^{(N)}_j\right)<N/4$ from any Hadamard codeword is said
to be within the $t-$error-correcting capability of the Hadamard
code~\cite{Hor07}.

In our analysis, we introduce Hadamard codewords with two types of
bit errors: unrestricted errors and restricted errors.
Unrestricted bit errors can occur at any of the~$N$ bit positions,
whereas restricted errors are limited to~$N/2$ specific bit
positions.

\subsubsection{Codewords with unrestricted errors}

The codewords having unrestricted errors are the strings having
Hamming distance~$d$ from any Hadamard codeword $W^{(N)}_j$. We
define the set of codewords with errors specified with respect to
any particular codeword through the use of an error syndrome,
which represents all the possible ways an error of $d$ bits can
occur.

The error syndrome for $d$ unrestricted errors is
\begin{align} U_d=\left\{z\in \{0,1\}^N \mid |z|=d\right\}.\end{align}
The set of codewords having $d$ unrestricted errors with respect
to the~\nth{j} codeword is
 \begin{align}\label{Eq:unrestrictederrorsOne}
 \Xi^{(N)}_{j,d}=\left\{ z \oplus W_j^{(N)}  \mid  z \in U_d
 \right\},
\end{align} and the set of all codewords with zero to
$N/16$ unrestricted errors is
\begin{align}
\Xi^{(N)}_{j,\,\mathbb{Z}_{N/16}}=&\left\{\Xi^{(N)}_{j,m}\mid m\in
\mathbb{Z}_{N/16}\right\}.\label{Eq:unrestrictederrorsAll}
\end{align} We proceed in a similar manner with the definition of the
 Hadamard codewords having restricted errors.

\subsubsection{Codewords with restricted errors}

The codewords having restricted errors are the strings with
Hamming distance~$d$ from Hadamard codeword $W^{(N)}_j$, but the
errors are restricted to the $N/2$ specific bit positions where
the codeword $W^{(N)}_{N-1}$ contains a one. The error syndrome
for $d$ \emph{restricted} errors is
\begin{align}
R_d=\left\{z\in \{0,1\}^N \mid |z|=d \text{ and } z \preccurlyeq
W^{(N)}_{N-1}\right\}.
\end{align}
We say that a vector $a \in \{0,1\}^N$ is dominated by a vector $b
\in \{0,1\}^N$, denoted $a \preccurlyeq b$, if whenever $a_i =1$
then also $b_i = 1$. The set of codewords having $d$ restricted
errors with respect to the~\nth{j} codeword is
 \begin{align}\label{Eq:restrictederrorsOne}
 \tilde{\Xi}^{(N)}_{j,d}= \left\{ z \oplus W_j^{(N)}  \mid  z \in
 R_d\right\}.
\end{align} We present two examples of the sets
given by Eq.~(\ref{Eq:restrictederrorsOne}).

For the $N=8$ case where there is a single restricted error, we
have $W^{(8)}_{7}=01101001$, and, for the particular codeword
$W^{(8)}_{4}=00001111$,
\begin{align}\label{eq:expd1string}
\tilde{\Xi}^{(8)}_{4,1}=&\left\{01001111,00101111,00000111,00001110\right\}.
\end{align}
Inspection of the set given in Eq.~(\ref{eq:expd1string}) reveals
the error alignment with the bit positions where $W^{(8)}_{7}=1$.

For $N=8$ where there are two restricted errors, the errors may
occur at any two of four possible possible bit positions
represented as~$\{a,b,c,d\}$, for which there are the ${4 \choose
2}=6$ distinct bit error pairings $\{ab,ac,ad,bc,bd,cd\}$. The
codewords with two errors in this case are
\begin{align}\label{eq:expd2string}
\tilde{\Xi}^{(8)}_{4,2}=&\{00000110,00100111,00101110,01000111, \nonumber\\
&01001110,01101111\}.
\end{align}For the general case, the set having $m$-tuple restricted errors
has size~$\left|\tilde{\Xi}^{(N)}_{j,m}\right|={N/2 \choose m}$.

The set of all correctable codewords with zero to $N/4$ restricted
errors with respect to the $j^{\rm th}$ Hadamard codeword is
\begin{align}
\tilde{\Xi}^{(N)}_{j,\,\mathbb{Z}_{N/4}}=&\left\{\tilde{\Xi}^{(N)}_{j,m}\mid
m\in \mathbb{Z}_{N/4}\right\}.\label{Eq:restrictederrorsAll}
\end{align} The size of this set is exponential in $N$ since
\begin{align}
\left|\tilde{\Xi}^{(N)}_{j,\,\mathbb{Z}_{N/4}}\right|=&\sum_{m=0}^{N/4-1}{\frac{N}{2} \choose m}\nonumber\\
=&\,\frac{1}{2}\left[2^{\frac{N}{2}}-{\frac{N}{2} \choose
\frac{N}{4}}\right]\label{Sec2SizeNeary}.
\end{align} We give two variations of the close Hadamard problem in terms of Definition \ref{def:1}.

\begin{problem}\sl\label{restrictedCHP}
Given the set of codewords
$\tilde{A}=\tilde{\Xi}^{(N)}_{N/2-1,\,\mathbb{Z}_{N/4}}$, which
contains strings that are close (in the restricted sense) to the
Hadamard codeword $W_{N/2-1}^{(N)}$ and the set of codewords
$\tilde{B}=\tilde{\Xi}^{(N)}_{k,\,\mathbb{Z}_{N/4}}$, which
contains strings that are close (in the restricted sense) to any
other Hadamard codeword $W_k^{(N)}$ with $k \in
\mathbb{Z}_{N/2-1}$ and a string $z$ randomly selected with
uniform distribution~$\mu$ such that $z\in_\mu \tilde{C}=\tilde{A}
\cup \tilde{B}$, the \textbf{restricted close Hadamard problem} is
to determine whether $z\in \tilde{A}$ or $z\in \tilde{B}$ with the
fewest oracle queries.
\end{problem}

In our formulation of Problem~\ref{restrictedCHP}, we have made a
technical assumption by setting $j=N/2-1$ in our definition of
set~$\tilde{A}$. We could have selected any other $j \in
\mathbb{Z}_{N/2}$ as long as we excluded the selection from the
definition of set~$\tilde{B}$. We make this assumption because our
quantum algorithm requires the measurement of some qubit. We have
arbitrarily, and without loss of generality, set it to the qubit
that corresponds to~$j=N/2-1$. The same assumption is made in our
formulation of Problem~\ref{unrestrictedCHP}.

\begin{problem}\sl\label{unrestrictedCHP}
Given the set of codewords
$A=\Xi^{(N)}_{N/2-1\,\mathbb{Z}_{N/16}}$, which contains strings
that are close (in the unrestricted sense) to the Hadamard
codeword $W_{N/2-1}^{(N)}$ and the set of codewords
$B=\Xi^{(N)}_{k,\,\mathbb{Z}_{N/16}}$, which contains strings that
are close (in the unrestricted sense) to any other Hadamard
codeword $W_k^{(N)}$ with $k \in \mathbb{Z}_{N/2-1}$ and a string
$z$ randomly selected with uniform distribution~$\mu$ such that
$z\in_\mu C=A \cup B$. The \textbf{unrestricted close Hadamard
problem} is to determine whether $z\in A$ or $z\in B$ with the
fewest oracle queries.
\end{problem}

\subsection{Quantum Algorithm for Oracle Decision Problems}

Fig.~\ref{fig:SVDJCircuit} is a graphical representation of the
operators applied to the solution of
\begin{figure}[tbp]
            \begin{center}
            \includegraphics[width=9 cm]{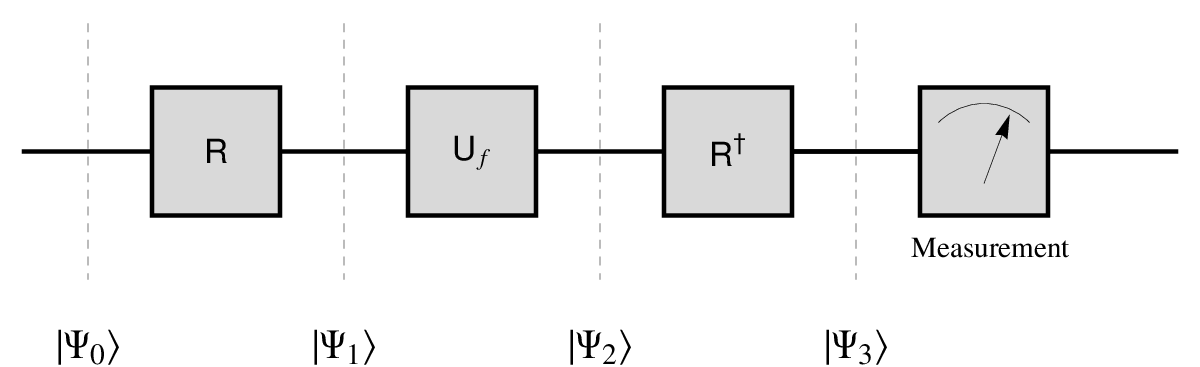}
            \end{center}
            \caption{Graphical representation of the operators for the solution of oracle decision problems
            in continuous variable
            settings~\cite{AHS09}. } \label{fig:SVDJCircuit}
\end{figure}
the Deutsch--Jozsa oracle decision problem employing logical
states encoded in the infinite-dimensional Hilbert space of the
harmonic oscillator~\cite{AHS09,AHS12}. In this case, $\Psi(x)$ is
a square-integrable function of the continuous position~$x$.

A key aspect of the approach is the physical accessibility of the
harmonic oscillator ground state and the availability of linear
and quadratic operators that enable us to prepare the logical
input state $\left|\Psi_0\right\rangle$. Additionally, the
operators $\op{R}$ and $\op{R}^\dag$ identified in
Fig.~\ref{fig:SVDJCircuit} are the easily implementable continuous
Fourier transform and its inverse. The oracle information is
encoded in the logical state by the unitary operator~\cite{AHS09}
\begin{equation} \op{U}_f=\left(\begin{array}{cccc}
(-1)^{z_1} &          0              &      \cdots     & 0      \\
      0            &   (-1)^{z_2}    &      \cdots     & 0      \\
     \vdots        &       \vdots            &      \ddots     & \vdots  \\
      0            &          0              &      \cdots           & (-1)^{z_N} \\
\end{array} \right)\label{Eq:Oracle_Uf},
\end{equation} where the $z_i$ are the bits of the unknown string given in oracle decision problem Definition 1.

Mapping this algorithm to the finite-dimensional,
continuously-parameterized coherent spin system must first deal
with the step that takes the ground state of the spin system to
the logical input state $\left|\Psi_0\right\rangle$. This first
step also employs physically accessible, linear spin rotation and
quadratic spin squeezing.


\section{The Spin System Model}

The method of generalized coherent states has been successfully
used to describe a number of diverse physical phenomena including
quantum optics, atom-light interactions, and
superfluidity~\cite{Pe72}. Here we make use of coherent spin
states~\cite{Ra71,Ar72} in creating an alternative model of
continuous variable quantum computation. Just as squeezing is
beneficial in continuous variable quantum computing using the
harmonic oscillator, we make use of spin squeezing
here~\cite{KU93}. We use the optimally squeezed spin
state~\cite{KU93} as input to our algorithm and show that it can
be approximated by a superposition of two discrete states with
constant error independent of the size of the Hilbert Space.

\subsection{Coherent Spin States}

Our spin system is a collection of $2S$ elementary $1/2$ spins.
$2S$ is an odd integer, and we choose $2S+1=N$ so that $N=2^n$-bit
strings may be naturally represented. We refer this as an $S$-spin
system~\cite{KU93}.

The system dynamics are determined by the Hamiltonian, which is
expressed as a polynomial of su(2) algebraic elements. These
algebraic elements are Pauli spin operators in the spin-1/2
single-particle case. For higher even dimensions, we use notation
similar to~\cite{KU93} with operators $\hat{\op{S}}_i$,
$\hat{\op{S}}_j$ and $\hat{\op{S}}_k$ and $i,j,k$ denoting the
components of any three orthogonal directions, such that
\begin{align}\label{Eq:SpinSysDef1}\left[\hat{\op{S}}_i,\hat{\op{S}}_j\right]&={\rm{i}}\hat{\op{S}}_k,
\end{align}
and
\begin{align}\label{Eq:SpinSysDef1a}
\Delta \hat{\op{S}}^2_i \Delta \hat{\op{S}}^2_j &\geq
\frac{1}{4}\left\langle \hat{\op{S}}_k\right\rangle^2,
\end{align} and cyclic permutations.

The spin system is oriented in the usual way. With
\begin{align}m\in
\left\{-s,-s+1,-s+2,\ldots,s\right\},\end{align}
 the spin kets $\left|m\right\rangle_s$ are eigenstates of $\hat{\op{S}}_z$ and
$\op{S}^2$ satisfying
\begin{align}\label{Eq:SpinSysDef2}
\hat{\op{S}}_z\left|m\right\rangle_s&=m\left|m\right\rangle_s,\end{align}
and
\begin{align}\op{S}^2\left|m\right\rangle_s&=s(s+1)\left|m\right\rangle_s,
\end{align} where $\op{S}^2=\hat{\op{S}}_x^2+\hat{\op{S}}_y^2+\hat{\op{S}}_z^2$.
The ladder operators  are $\hat{\op{S}}_\pm=\hat{\op{S}}_x\pm{\rm
i}\hat{\op{S}}_y$, and the action of the lowering operator on the
ground state is
\begin{align}\label{Eq:SpinSysDef4}
\hat{\op{S}}_-\left|-s\right\rangle_s=0.
\end{align} We use the discrete spin states to construct the continuously parameterized coherent spin states.

The harmonic oscillator coherent states are translations of the
oscillator ground state~\cite{Le97}. Analogously, the coherent
spin states are rotations of the spin system ground
state~\cite{KU93,Ar72,Pe72}. Individual spin states are often
referred to in the literature \cite{Ar72,Pe72} as Dicke states
analogous to photon number states, and the coherent spin states
are referred to as Bloch states analogous to Glauber states.

The coherent spin state,~$\left|\theta,\phi\right\rangle_{s}$ with
~$\theta,\phi \in \mathbb{R}$, is~\cite{KU93}:
\begin{align}\label{Eq:CSS_Rotation_Op}
\left|\theta,\phi\right\rangle_{s}&=\op{R}_{\theta,\phi}\left|-s\right\rangle_s
\nonumber\\
&=\left(1+\tan^2\frac{\theta}{2}\right)^{-s}\times\nonumber\\
&\phantom{=}\,\sum_{k=0}^{2s}{2s\choose
k}^{\frac{1}{2}}\left(e^{{\rm{i}}\phi}\tan{\frac{\theta}{2}}\right)^k\left|s-k\right\rangle_s.
\end{align} The coherent spin state of most interest is
\begin{align}
\left|\pi/2,0\right\rangle_{s}=2^{-s}&\sum_{k=0}^{2s}{2s\choose
k}^{\frac{1}{2}}\left|s-k\right\rangle_s,\label{Eq:Sec3CSS}
\end{align}
which has a Dicke-state amplitude spectrum whose squared magnitude
is the binomial probability distribution with $p=q=1/2$ shown in
Fig.~\ref{fig:SVDJFig2QPlots}b.

Quasi-probability distributions~\cite{Pe72} are a useful means of
visualizing spin states. We choose to use Q-functions~\cite{Pe72}
as coherent-state representations. The spherical plots of these
distributions provide good intuition as to the orientation and the
isotropic or anisotropic distribution of uncertainties, but they
are not used to actually calculate uncertainties, which is
performed using Eq.~(\ref{Eq:SpinSysDef1a}).

For the arbitrary coherent spin state represented as
$\left|\Psi\right\rangle=\sum_{k=0}^{2s}
\alpha_{k}\left|k\right\rangle$, we express the spherical
Q-function \cite{Le97} as
\begin{align}
Q(\theta,\phi)=&\sum_{k=0}^{2s}{2s\choose
k}^{\frac{1}{2}}\sin(\theta/2)^k \cos(\theta/2)^{2s-k}\alpha_{k}
e^{{\rm{i}} k \phi}.\label{Sec3QF}
\end{align}
\begin{figure}[tbp]
            \begin{center}
            \includegraphics[width=9 cm]{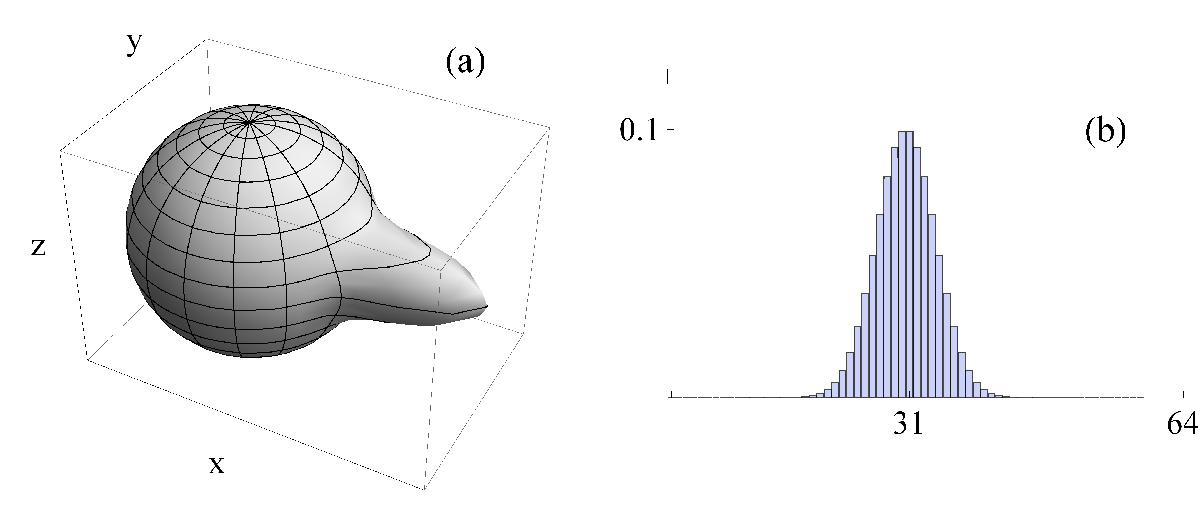}
            \end{center}
            \caption{(a) Spherical
            Q-function of the state given by
Eq.~(\ref{Eq:Sec3CSS}) for $s=\frac{63}{2}$, and (b) Plot of the
respective Dicke-state probability distribution.}
\label{fig:SVDJFig2QPlots}
\end{figure}In Fig.~\ref{fig:SVDJFig2QPlots}, we plot the spherical Q-function
and the probability distribution of the discrete spin state for
the state given in Eq.~(\ref{Eq:Sec3CSS}). Note that this coherent
spin state appears as an `equatorial' state with isotropic
uncertainty distribution when represented this way.

\subsection{Squeezed Spin States}
Coherent spin states can be squeezed \cite{KU93}. Whereas the
Glauber states can be squeezed to an arbitrary degree, spin states
can only be squeezed to the Heisenberg limit of $1/2$~\cite{KU93}.
We wish to exploit the squeezed state with the minimal achievable
variance in our algorithm. In the following, we formulate
expressions for this optimally squeezed spin state and show that
it can be approximated well by a superposition of two spin states.

We employ two-axis counter-twisting \cite{KU93} to define the
squeezing operator
\begin{align}
\op{S}_\mu=&e^{\rm{i} \frac{\pi}{4}\hat{\op{S}}_x}e^{\rm{i}\mu
\left(\hat{\op{S}}^2_z-
\hat{\op{S}}^2_y\right)},\label{eq:ShatSqOp}
\end{align} where $\mu$ is the squeezing
parameter~\cite{KU93}. The rotation operator $e^{\rm{i}
\frac{\pi}{4}\hat{\op{S}}_x}$ orients the resulting anisotropic
uncertainty distribution in the $y,z$ directions.

Applying the operator $\op{S}_\mu$ to
\begin{align}\left|\Psi\right\rangle=\left|\pi/2,0\right\rangle_{s}\label{Eq:SimplePsi}
\end{align} allows us to reduce the variance $\Delta \hat{\op{S}}^2_z$ at the expense of enhancing the variance $\Delta
\hat{\op{S}}^2_y $. The reduced variance may be expressed as
\begin{align}
V_{-}= \langle \hat{\op{S}}^2_z \rangle= \left\langle\Psi\right|
\op{S}^\dag_\mu\hat{ \op{S}}^2_z \op{S}_\mu
\left|\Psi\right\rangle \label{Eq.Vminus}
\end{align}since the first moment $\langle \hat{\op{S}}_z \rangle=0$. In Fig. \ref{fig:SVDJFig3SqueezedQPlots}a,
we plot the quasi-probability distribution of a squeezed spin
state. The reduced variance of the squeezed state in the $z$
direction and increased variance in the $y$ direction is evident.
\begin{figure}[tbp]
            \begin{center}
            \includegraphics[width=9 cm]{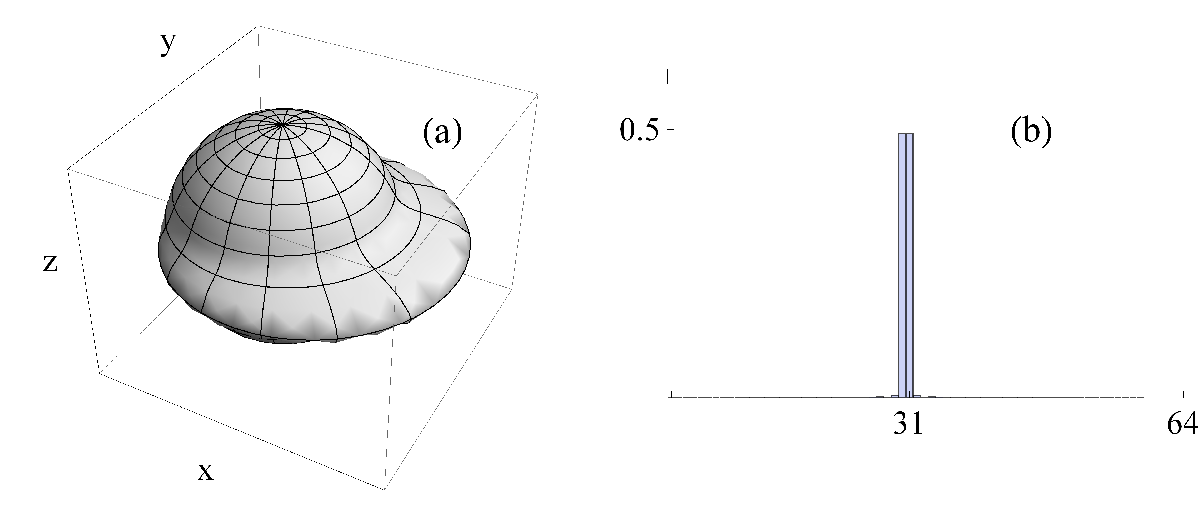}
            \end{center}
            \caption{(a) Spherical Q-function of the
            squeezed state given by Eq.~(\ref{Eq:PhiN_1}) for $s=\frac{63}{2}$, and (b)
            Plot of the respective Dicke-state probability distribution.
            } \label{fig:SVDJFig3SqueezedQPlots}
\end{figure}

The minimum value of the reduced variance $V_{-}$ asymptotically
approaches 1/2 with increasing $s$ \cite{KU93}. We refer to the
optimal value of the squeezing parameter at this minimum as
$\mu_{\text{opt}}$. For $\mu>\mu_{\text{opt}}$, the distribution
variance increases and the distribution quasi-probability
distribution becomes skewed ~\cite{KU93}. It can shown be that
$\mu_{\text{opt}} \rightarrow \frac{1}{s}$ as
$s\rightarrow\infty$. This limit is understandable since the
variance of a binomial distribution with $p=q=1/2$ is $N/4$, and
squeezing simply has the effect of removing the distribution
variance of the dependency on $N=2s+1$.

We express the optimally squeezed spin state as
\begin{align}\label{Eq:PhiN_1}
\left|\Phi^{(N)}\right\rangle&=\left|\pi/2,0,\mu_{\text{opt}}\right\rangle_s\nonumber\\
&=\op{S}_{\mu_{\text{opt}}}\left|\Psi\right\rangle,
\end{align} with $\left|\Psi\right\rangle$ defined in Eq.~(\ref{Eq:SimplePsi}).
In Fig. \ref{fig:SVDJFig3SqueezedQPlots}b, we plot the Dicke-state
probability distribution of the optimally squeezed state. It is
evident that this state approximates the superposition of two spin
states. We wish to provide a bound on how well approximated the
squeezed state is by a two-component superposition.

Analysis of the variance of the squeezed state's probability
distribution is facilitated using the qudit representation rather
than the spin state representation. We represent this
$N$-dimensional squeezed state in terms of the qudits $|i\rangle$
as
\begin{align}\label{Eq:PhiOutAmpDef}
\left|\Phi^{(N)}\right\rangle=\sum_{i=0}^{N-1} \alpha_i |i\rangle.
\end{align}
The probability distribution associated with
$\left|\Phi^{(N)}\right\rangle$ may be represented as the set
\begin{align}\label{Eq:PhiOutProbDef}
\mathcal{P}^{(N)}=\{\left|\alpha_0\right|^2,\ldots,\left|\alpha_i\right|^2,\ldots,\left|\alpha_{N-1}\right|^2\},
\end{align}
with individual probabilities
$\mathcal{P}^{(N)}_i=\left|\alpha_i\right|^2$. We note that the
squeezed state is symmetric about the centre, and the two central
states have
\begin{align}\label{Eq:ProbCentreComp}
\mathcal{P}^{(N)}_{(N/2-1)}&=\mathcal{P}^{(N)}_{(N/2)}=P_c^{(N)},
\end{align} and thereby form the principle
components of the probability distribution of optimally squeezed
states.

For $s>\frac{3}{2}$, the expression for the reduced variance given
by Eq.~(\ref{Eq.Vminus}) requires solving eigenvalue problems of
degree greater than eight and is no longer analytic, and we must
resort to numerical analysis. For $s=3/2$ ($N=4$), the expression
for the reduced variance given by Eq.~(\ref{Eq.Vminus}) is
analytic, and $\mu_{\text{opt}}=\frac{\pi}{6\sqrt{3}}$. We
represent this optimal squeezed state as
\begin{align}\label{Eq:Optimal2SEquals3}
\left|\Phi^{(4)}\right\rangle=e^{\rm{i}\phi}\left(0|0\rangle+\frac{1}
{\sqrt{2}}|1\rangle+\frac{1}{\sqrt{2}}|2\rangle+0|3\rangle\right),
\end{align} where $e^{\rm{i}\phi}$ is a global phase picked up by
the action of $\op{S}_\mu$. The associated probability
distribution is
\begin{align}
\mathcal{P}^{(4)}=\left\{0, 1/2,1/2,0\right\}.
\end{align} For this case we achieve
what we refer to as `perfect' squeezing, where `perfect' means
that the two central components have probability equal to a half,
and the probability of the other two components is zero.

However, this four-component distribution has a variance of only a
quarter, where the distribution variance is expressed as
\begin{align}\label{Eq:PhiOutProbDistVar}
\text{Var}\left[\mathcal{P}^{(N)}\right]=\sum_{i=0}^{N-1} i^2
\left|\alpha_i\right|^2-\left(\sum_{i=0}^{N-1}i
\left|\alpha_i\right|^2\right)^2.
\end{align} Indeed, for all $N$ if $\mathcal{P}_c^{(N)}=1/2$ then,
\begin{align}
\text{Var}\left[\mathcal{P}^{(N)}\right]=&\frac{1}{2}\left((N/2-1)^2+(N/2)^2\right)-\frac{1}{4}\left(N-1\right)^2\nonumber\\
=&\frac{1}{4}.
\end{align} Since the distribution variance approaches 1/2 as $N=2s+1$ approaches
infinity, perfect squeezing in the sense we have defined is not
possible. We use the variance equals 1/2 as a means to bound
$\mathcal{P}_c^{(N)}$ defined in Eq.~(\ref{Eq:ProbCentreComp}). In
Fig.~\ref{Fig:VarianceandProb}a,
\begin{figure}[tbp]
            \begin{center}
            \includegraphics[height=8 cm]{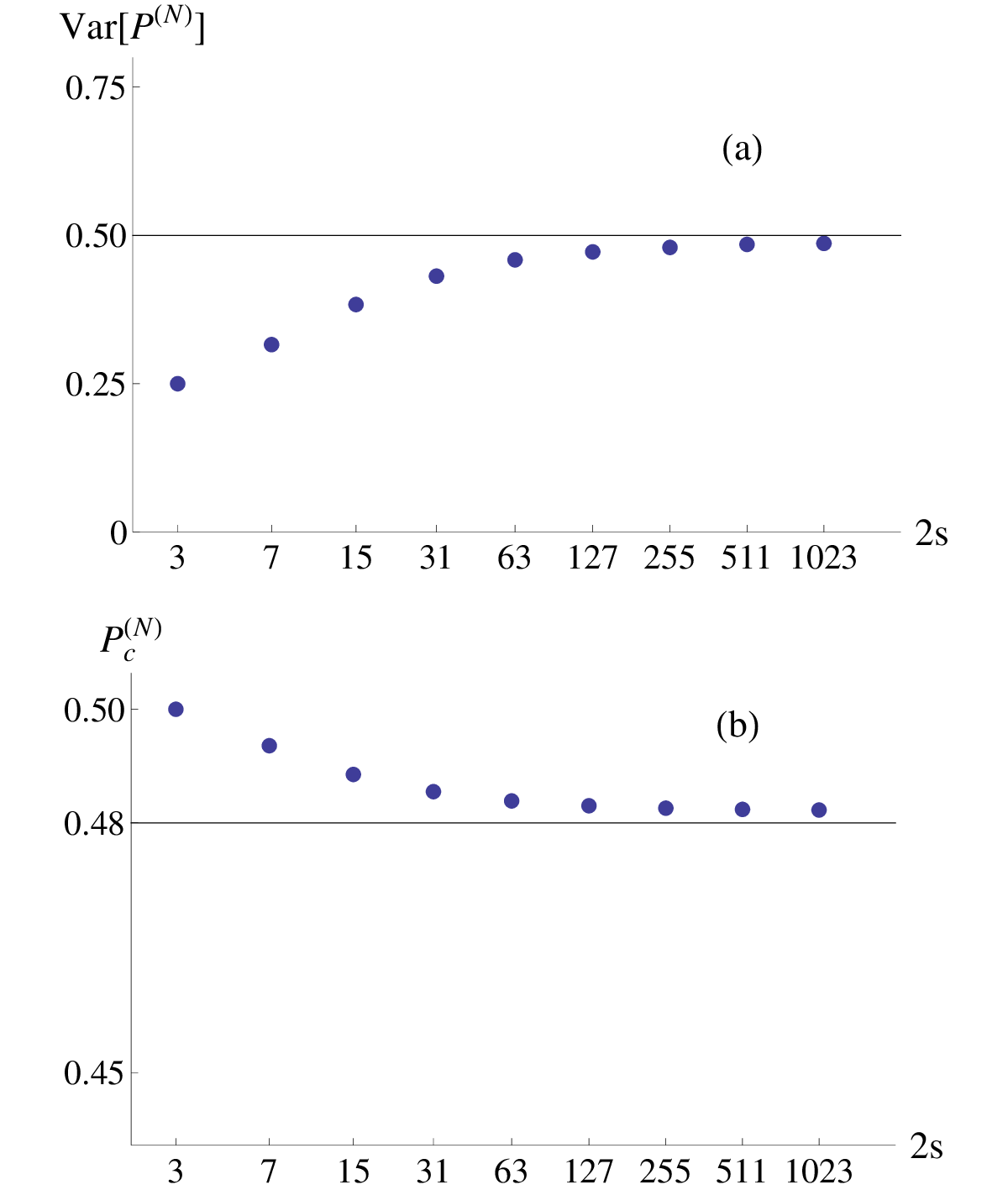}
            \end{center}
            \caption{(a) Calculated value of the reduced variance of the
probability distribution given by Eq.~(\ref{Eq.Vminus}) approaches
$1/2$ with increasing $s$ as predicted~\cite{KU93}. (b) Calculated
value of the probability of the two central components given by
Eq.~(\ref{Eq:ProbCentreComp}) is bounded by the constant given by
Eq.~(\ref{Eq:CentralProbBound}).} \label{Fig:VarianceandProb}
\end{figure}
we plot the calculated values of the reduced variance given by
Eqs.~(\ref{Eq.Vminus}) and~(\ref{Eq:PhiOutProbDistVar}) as a
function of $s$ from $s=3/2$ to $s=1023/2$, where we observe that
the variance approaches $V=1/2$ as predicted. In order to bound
the limiting value of the two central components
$\mathcal{P}_c^{(N)}$, we bound the `tails' of the probability
distribution $\mathcal{P}^{(N)}$.

In Fig.~\ref{Fig:TheTails}, we plot histograms calculated from the
squeezed distribution,~$\mathcal{P}^{(N)}$, for several values of
$s$, where we have scaled the ordinate to reveal the structure of
the tail components. We see that the components immediately
adjacent to the central components have
\begin{align}
\mathcal{P}_{N/2+1}^{(N)}=\mathcal{P}_{N/2-2}^{(N)}\approx0,
\end{align} and further outlying terms tail off in an exponential-like fashion.
We thus introduce the following bounding probability distribution
\begin{align}\label{Eq:PhiOutProbDef_best}
\mathcal{P}_{\mathcal{B}}^{(N)}=&\left\{0,\ldots,0,\frac{\epsilon}{3},
\frac{2\epsilon}{3},0,\frac{1}{2}-\epsilon,\right.\\\nonumber
&\,\,\,\left.\frac{1}{2}-\epsilon,0,\frac{2\epsilon}{3},\frac{\epsilon}{3},0,\ldots,0\right\},
\end{align} in order to calculate a bound on the two central components of the distribution.

Solving
$\text{Var}\left[\mathcal{P}_{\mathcal{B}}^{(N)}\right]=1/2$ for
$\epsilon$ gives the probability of the two central components
\begin{align}\label{Eq:CentralProbBound}
\mathcal{P}_{\mathcal{B}_c}^{(N)}=\frac{1}{2}-\epsilon < 0.484.
\end{align} In Fig.~\ref{Fig:VarianceandProb}b, we plot the calculated values
of $\mathcal{P}_c^{(N)}$, where we note that it goes from 1/2 at
$s=3/2$ and asymptotically approaches the constant bounded from
below by Eq.~(\ref{Eq:CentralProbBound}). The bounding
distribution is also overlayed on the histograms presented in
Fig.~\ref{Fig:TheTails}.

With these concepts behind us, we state a theorem.

\begin{theorem}\label{CHPinSS}The close Hadamard problem, as defined by Problem~\ref{restrictedCHP}
and Problem~\ref{unrestrictedCHP}, can be solved in the spin
system model with arbitrarily small error probability in a
constant number of oracle queries.
\end{theorem}

\begin{figure}[tbp]
            \begin{center}
            \includegraphics[width=9 cm]{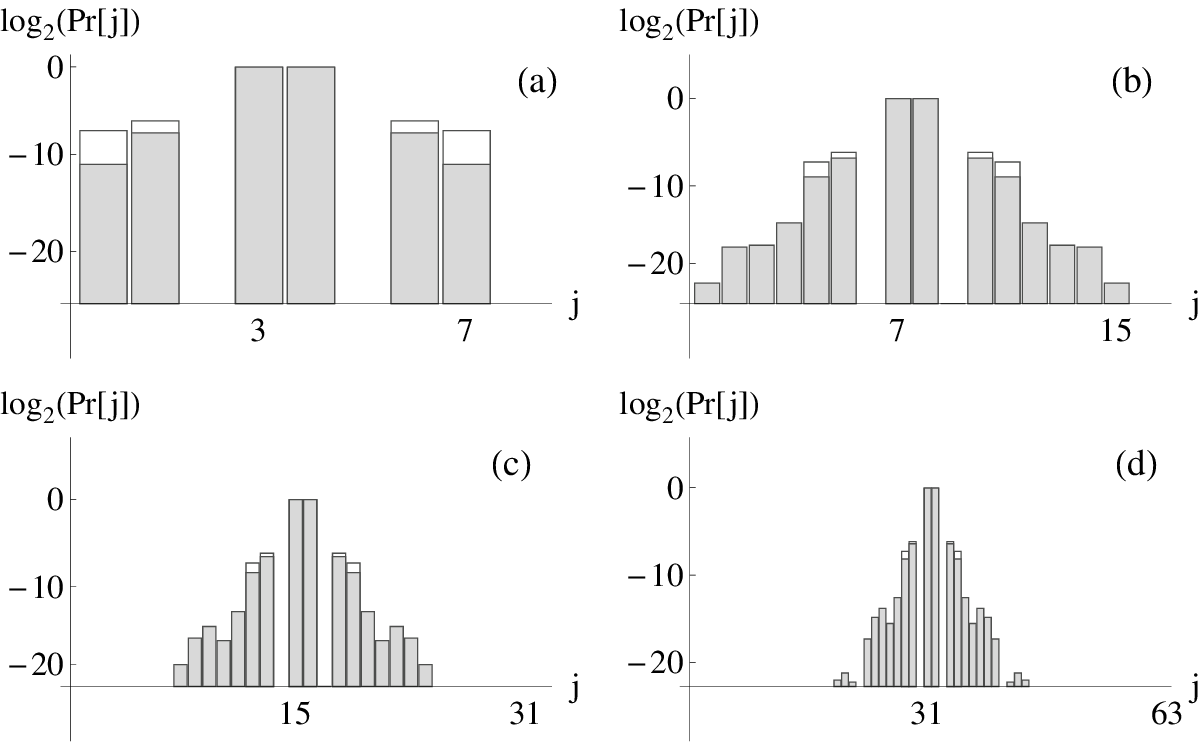}
            \end{center}
            \caption{Histograms of the  probability distribution $\mathcal{P}^{(N)}$
            for (a)~$s=7/2$, (b)~$s=15/2$, (c)~$s=31/2$ and (d)~$s=63/2$ with a logarithmic scale
            for the ordinate. The bounding distribution given by
            Eqs.~(\ref{Eq:PhiOutProbDef_best}) and~(\ref{Eq:CentralProbBound})
            is overlayed on each of the histograms.}
\label{Fig:TheTails}
\end{figure}

Here we briefly recap and summarize the overall approach before
proving the theorem in the next section. The spin system model
gives us the optimally squeezed state
$\left|\Phi^{(N)}\right\rangle$ given by
Eq.~(\ref{Eq:PhiOutAmpDef}), which forms the input state to the
algorithm operators graphically represented in
Fig.~\ref{fig:SVDJCircuit}. The input state
$\left|\Phi^{(N)}\right\rangle$ has an associated probability
distribution $\mathcal{P}^{(N)}$ given by
Eq.~(\ref{Eq:PhiOutProbDef}) and bounded by the distribution
$\mathcal{P}_{\mathcal{B}}^{(N)}$  given by
Eq.~(\ref{Eq:PhiOutProbDef_best}).

In order to facilitate analysis, we approximate the optimally
squeezed state with an idealized input state consisting of the
only the two central components. Since greater than $96\%$ of the
probability is manifest in the two central components of the
bounding distribution as given by Eq.~(\ref{Eq:CentralProbBound}),
we can replace the optimally squeezed state with the idealized
superposition of two states and fold the approximation into the
resulting single-query success probability.

This idealized input state can be expressed in a variety of ways.
In the qudit representation, it is written
\begin{align}\label{Eq:idealqudit}
\left|\Phi^{(N)}_{\text{\tiny{ideal}}}\right\rangle=&0|0\rangle+...+\frac{1}{\sqrt{2}}|N/2-1\rangle+\nonumber\\
&\frac{1}{\sqrt{2}}|N/2\rangle+...+0|N-1\rangle.
\end{align} In the spin representation with $s=\frac{N-1}{2}$, it is written
\begin{align}\label{Eq:idealspin}
\left|\Psi^{(s)}_{\text{\tiny{ideal}}}\right\rangle=&0|{-}s\rangle+0|{-s}+1\rangle+...+\frac{1}{\sqrt{2}}|{-}1/2\rangle+\nonumber\\
&\frac{1}{\sqrt{2}}|1/2\rangle+...+0|s-1\rangle+0|s\rangle.
\end{align}
Finally in shorthand spin representation, the idealized input
state is simply written
\begin{align}\label{Eq:idealized_input}
|\Psi_0\rangle=\frac{1}{\sqrt{2}}\left(\left|\frac{1}{2}\right\rangle_s
+\left|-\frac{1}{2}\right\rangle_s\right),\end{align}where we have
used the subscript zero to represent the algorithm input state in
Fig.~\ref{fig:SVDJCircuit}. Proof analysis proceeds by assuming we
have a quantum algorithm employing the operators
$\op{R}\op{U_f}\op{R}^\dag$ given in Fig.~\ref{fig:SVDJCircuit}
with $\op{R}=\op{R}^\dag=\op{H}^{\otimes n}$ acting on the
idealized input state given by Eq.~(\ref{Eq:idealized_input}).


\section{Algorithms for the Close Hadamard Problem }

In this section we prove Theorem~\ref{CHPinSS} demonstrating what
we have an algorithm that solves the close Hadamard problem in the
spin system model. We conclude this section with a discussion
bounding the performance of classical algorithms in the solution
of the close Hadamard problem.

\subsection{Bounding the Performance of the Close Hadamard Algorithm in the Spin System Model}

To facilitate analysis in the spin system case, we first use the
idealized state given by Eq~(\ref{Eq:idealized_input}), and then
reintroduce the effect of the probability distribution of the
optimally squeezed state at the end. We proceed by breaking the
proof of Theorem~\ref{CHPinSS} up into two claims one for the
restricted case and one for the unrestricted case.

We suppress the normalization factor~$\frac{1}{\sqrt{2}}$ and
express the idealized input state as
\begin{align}\label{Eq:AlgInput}
\left|\Psi_0\right\rangle=\left|\frac{1}{2}\right\rangle_s+
\left|-\frac{1}{2}\right\rangle_s.\end{align} We will suppress
this normalization factor in all equations having a simple
two-component superposition in order to make the equations easier
to read. The action of the algorithm on the input state is
expressed as
\begin{align}
\left|\Psi_3\right\rangle=&\op{H}^{\otimes
n}\op{U}_z\op{H}^{\otimes
n}\left|\Psi_0\right\rangle.\label{Sec2LmA1a}
\end{align} The $N$-bit string $z$ represents the function $f$.
We show that the algorithm efficiently solves both versions of the
close Hadamard problem.

\begin{claim}\label{restrictedclaim}The restricted close Hadamard problem can be solved with certainty
in a single oracle query using the quantum circuit given in
Fig.~\ref{fig:SVDJCircuit} and the idealized input state given by
Eq.~(\ref{Eq:AlgInput}).
\end{claim}

\begin{claim}\label{unrestrictedclaim}The unrestricted close Hadamard problem can be solved with
arbitrarily small error probability in a constant number of oracle
queries using the quantum circuit given in
Fig.~\ref{fig:SVDJCircuit} and the idealized input state given by
Eq.~(\ref{Eq:AlgInput}).
\end{claim}

A key feature of the input state is that it is a symmetric
superposition of two basis states. When Hadamard codewords are
encoded into the oracle, the action of the algorithm preserves the
symmetric superposition. This preservation is demonstrated in Fig.
\ref{Fig_for_Lemma1}.

\begin{figure}[tbp]
            \begin{center}
            \includegraphics[width=9 cm]{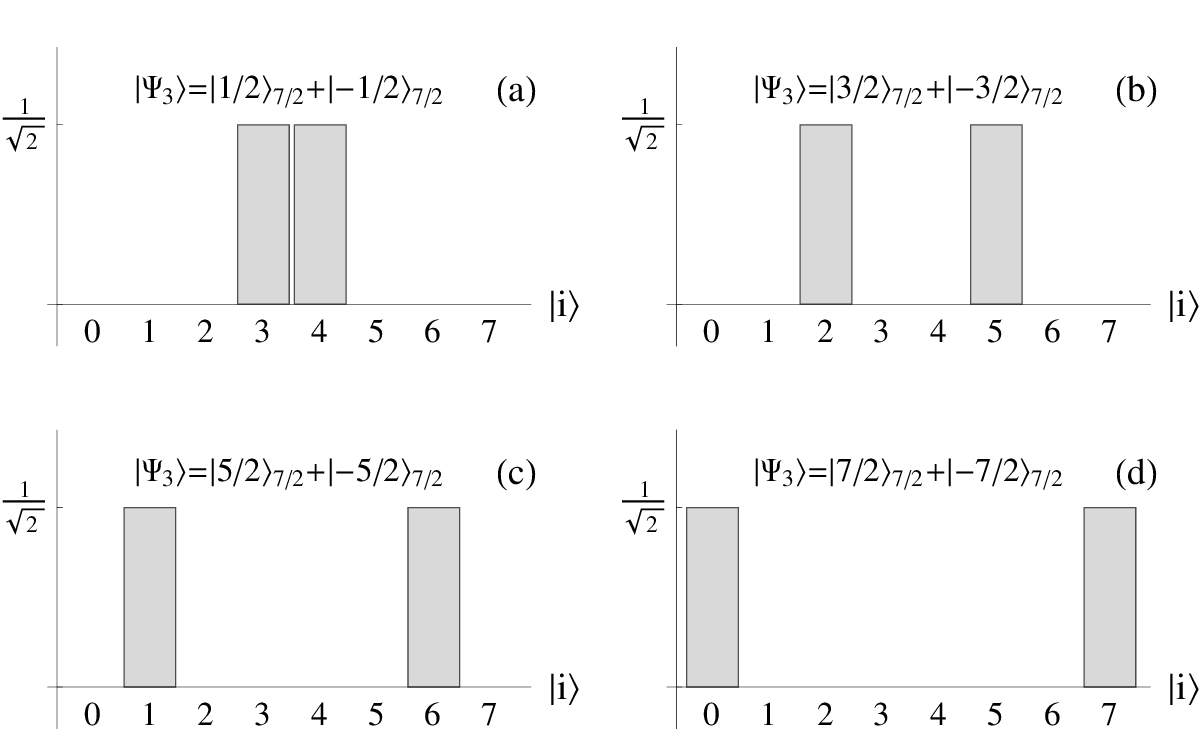}
            \end{center}
            \caption{For the example of $s=7/2$, $N=8$, the output state given by Eq.~(\ref{Sec2LmA1b})
            remains a symmetric superposition of two states for the Hadamard codewords: (a) $W^{(8)}_0$, (b)$W^{(8)}_1$,
             (c) $W^{(8)}_2$, and (d) $W^{(8)}_3$.}
\label{Fig_for_Lemma1}
\end{figure}

\begin{lemma}\sl\label{Lemma1}
Given the input
$\left|\Psi_0\right\rangle=\left|\frac{1}{2}\right\rangle_s+
\left|-\frac{1}{2}\right\rangle_s$ to the circuit shown in
Fig.~\ref{fig:SVDJCircuit} and the oracle encoded with one of the
Hadamard codewords $z= W_j^{(N)}$ for $N=2s+1$~and~$0\leq j <
\frac{N}{2}$, the output state is another superposition of spin
states
$\left|\Psi_3\right\rangle=\left|\frac{1}{2}+j\right\rangle_s+
\left|-\frac{1}{2}-j\right\rangle_s$.
\end{lemma}

\begin{proof} In order to
simplify notation in the following, we suppress the superscript
$N$ in $W_j^{(N)}$.  With $W$ defined as the set of Hadamard
codewords given by Eq.~(\ref{Sec1EqnLogH}), the pair $(W,\oplus)$
forms a group under addition modulo two~\cite{Hor07}. In
particular, the identity element is $W_0$, and each element is its
own inverse since
\begin{align}
W_j \oplus W_j=W_0.\label{Sec2SubSA2WgrpIdent}
\end{align}
It follows from the group property that the addition of any two
codewords is another codeword. For the Hadamard codeword pairs
$W_{j}$ and $W_{N-1-j}$, it can be readily shown that
\begin{align}
W_{j}\oplus
W_{N-1-j}=W_{N-1}\label{eq:HadamardCodeWordSumProperty}
\end{align} for all $j
\in \mathbb{Z}_N$.

With some algebraic manipulation, the state
$\left|\Psi_3\right\rangle$ may be expressed as
\begin{align}
\left|\Psi_{3}\right\rangle=&\frac{1}{N}\sum_{y=0}^{N-1}\left(\sum_{x=0}^{N-1}\alpha_{x,y}\right)
\left|y\right\rangle.\label{Sec2LmA1b}
\end{align}
We use the qudit representation $|y\rangle$ rather than the spin
state representation $|m\rangle_s$. We translate back to spin
state representation as the last step.

The symbol
\begin{align}
\alpha_{x,y}=&(-1)^{\left(z \oplus
W_{y+N/2}\right)_x}+(-1)^{\left(z \oplus
W_{N/2-(y+1)}\right)_x},\label{Sec2LmA1c}
\end{align}
where $z=W_j$ is the string encoded in the oracle, and the symbol
$x$ represents the $x^{\rm{th}}$ bit of the $N$-bit strings. Note
that the sums~$y+N/2$ and~$N/2-(y+1)$ in Eq.~(\ref{Sec2LmA1c}) are
modulo~$N$ sums.

The Hadamard codewords are balanced with the exception of $W_0$,
which is constant. For $j\neq k$, this allows us to write
\begin{align}
0=&\sum_{x=0}^{N-1}(-1)^{\left(W_{j}\oplus
W_k\right)_x},\label{eq:Lemma1a0Sum}
\end{align} and for $j=k$,
\begin{align}
N=&\sum_{x=0}^{N-1}(-1)^{\left(W_{j}\oplus
W_j\right)_x},\label{eq:Lemma1a1Sum}
\end{align} where we have used the group inverse relation
given in Eq.~(\ref{Sec2SubSA2WgrpIdent}). Using this result, we
see that a non-zero sum of $\alpha_{x,y}$ in Eq.~(\ref{Sec2LmA1b})
occurs exactly twice when~$y=(N/2+j)\mod N$ and
when~$y=(N/2-1-j)\mod N$.

In the qudit representation, the output state is expressed
\begin{align}
\left|\Psi_3\right\rangle=\left|N/2-1-j\right\rangle+\left|N/2+j\right\rangle,
\end{align} where qudit indices are understood to be modulo~$N$. For $0\leq j <N/2$, this translates back to the spin
basis as
\begin{align}
\left|\Psi_3\right\rangle=\left|-1/2-j\right\rangle_s+\left|1/2+j\right\rangle_s,
\end{align} thus completing the proof of Lemma 1.
\end{proof}

We now show that the superposition of two states is also preserved
for codewords with restricted errors.

\begin{lemma}\sl\label{Lemma2}
Given the input
$\left|\Psi_0\right\rangle=\left|\frac{1}{2}\right\rangle_s+
\left|-\frac{1}{2}\right\rangle_s$ to the circuit shown in
Fig.~\ref{fig:SVDJCircuit} and the oracle encoded with $z \in
\tilde{\Xi}^{(N)}_{j,\,\mathbb{Z}_{N/4}}$ given by
Eq.~(\ref{Eq:restrictederrorsAll}), which is a codeword having
less than $N/4$ restricted errors, and with $0\leq j
<\frac{N}{2}$, the output state is another superposition of spin
states
$\left|\Psi_3\right\rangle=\left|\frac{1}{2}+j\right\rangle_s+
\left|-\frac{1}{2}-j\right\rangle_s$.
\end{lemma}

\begin{proof}

Observe that the form of Eq.~(\ref{Sec2LmA1c}) allows for
cancelling of errors. Under certain conditions if an error occurs
at bit position $x$, the effect on the left-hand side of the plus
sign is cancelled by the opposite effect on the right-hand side.
Consider the $x^{\rm th}$ bit error in Eq.~(\ref{Sec2LmA1c}),
where we have cancellation
\begin{align}
0=&(-1)^{\left(z \oplus W_{y+N/2} \right)_x}+(-1)^{\left(z \oplus
W_{N/2-(y+1)} \right)_x}.\label{Sec2LmA2a}
\end{align}
This identity implies that
\begin{align}
1&=\left(W_{y+N/2}\oplus W_{N/2-(y+1)}\right)_x \nonumber\\
&=(W_{N-1})_x,\label{Sec2LmA22}
\end{align}
where we have used the result expressed in
Eq.~(\ref{eq:HadamardCodeWordSumProperty}). This is exactly the
same as the requirement to be a member of set
$\tilde{\Xi}^{(N)}_{j,\,\mathbb{Z}_{N/4}}$ defined in
Eq.~(\ref{Eq:restrictederrorsAll}). Under this condition, the
result of Lemma 1 holds
and~$\left|\Psi_3\right\rangle=\left|\frac{1}{2}+j\right\rangle_s+
\left|-\frac{1}{2}-j\right\rangle_s$.\\\end{proof}

We now show that the perfect superposition of two states is no
longer preserved for codewords having unrestricted errors. The
effect of errors that are not of the restricted type is to degrade
the superposition by distributing amplitude evenly across all
other states. However, as long as the number of these errors is
less than $N/16$, it is still possible to efficiently identify the
desired state.

\begin{lemma}\sl\label{Lemma3}
Given the input
$\left|\Psi_0\right\rangle=\left|\frac{1}{2}\right\rangle_s+
\left|-\frac{1}{2}\right\rangle_s$ to the circuit shown in
Fig.~\ref{fig:SVDJCircuit} and the oracle encoded with $z \in
\Xi^{(N)}_{j,\,\mathbb{Z}_{N/16}}$ given by
Eq.~(\ref{Eq:unrestrictederrorsAll}), which is a codeword having
less than $N/16$ unrestricted errors, the desired state can be
identified with success probability of at least $\frac{9}{16}$.
\end{lemma}

\begin{proof}
Let $V_j\in\Xi^{(N)}_{j,1}$ be a Hadamard codeword with a single
unrestricted error. We have already shown that if the error is of
the restricted type, two-component superpositions are preserved.
If the the error is not a restricted error, we have no error
cancellation, so the single bit error breaks the balanced and
constant sums defined by Eqs.~(\ref{eq:Lemma1a0Sum})
and~(\ref{eq:Lemma1a1Sum}), respectively. For $j\neq k$ this gives
\begin{align}
2=&\sum_{x=0}^{N-1}(-1)^{\left(W_{j}\oplus
V_k\right)_x},\label{eq:Lemma3a0Sum}
\end{align} and for $j=k$,
\begin{align}
N-2=&\sum_{x=0}^{N-1}(-1)^{\left(W_{j}\oplus
V_j\right)_x}.\label{eq:Lemma3a1Sum}
\end{align} As the input state is a two-component superposition,
the above sums result in all the amplitudes of the output state
either acquiring or losing an amount of amplitude proportional to
four---two from the amount in the balanced or constant sums given
in Eqs.~(\ref{eq:Lemma3a0Sum}) and~(\ref{eq:Lemma3a1Sum}) and two
from the effect of there being two components in the input state.

We express the output state for the worst case of a single
unrestricted error as
\begin{align}
\left|\Psi_3\right\rangle=&\frac{1}{\sqrt{2}}\left(1-\frac{4}{N}\right)\left(\left|\frac{1}{2}+j\right\rangle_s+
\left|-\frac{1}{2}-j\right\rangle_s\right)\nonumber\\
&+\frac{4}{\sqrt{2}N}\sum_{\stackrel{k=-s-1/2}{k\neq \pm
j}}^{k=s-1/2}\pm\left|\frac{1}{2}+k\right\rangle_s.
\end{align}
For the worst case of $l$ unrestricted errors, where no errors are
of the restricted type, the principle components have amplitude
\begin{align}\label{eq:compampprinciple}
\alpha=\frac{1}{\sqrt{2}}\left(1-\frac{4l}{N}\right),
\end{align} and the amplitude of the next largest component is
\begin{align}\label{eq:compampnext}
\beta=\frac{4l}{\sqrt{2}N}.
\end{align}
The amplitude reduction of the principle components by an amount
directly proportional to the number of errors results from the
constant sum given by Eq.~(\ref{eq:Lemma3a1Sum}) being reduced by
double the number of errors. However, the balanced sums are
variable since errors can cancel. The worst case occurs when the
errors are `in phase' resulting in the amplitude of the next
largest component being proportional to the number of errors. The
effect of codewords with unrestricted errors on the input
superposition is presented in Fig.~\ref{fig:comperrors}.

\begin{figure}[tbp]
            \begin{center}
            \includegraphics[width=9 cm]{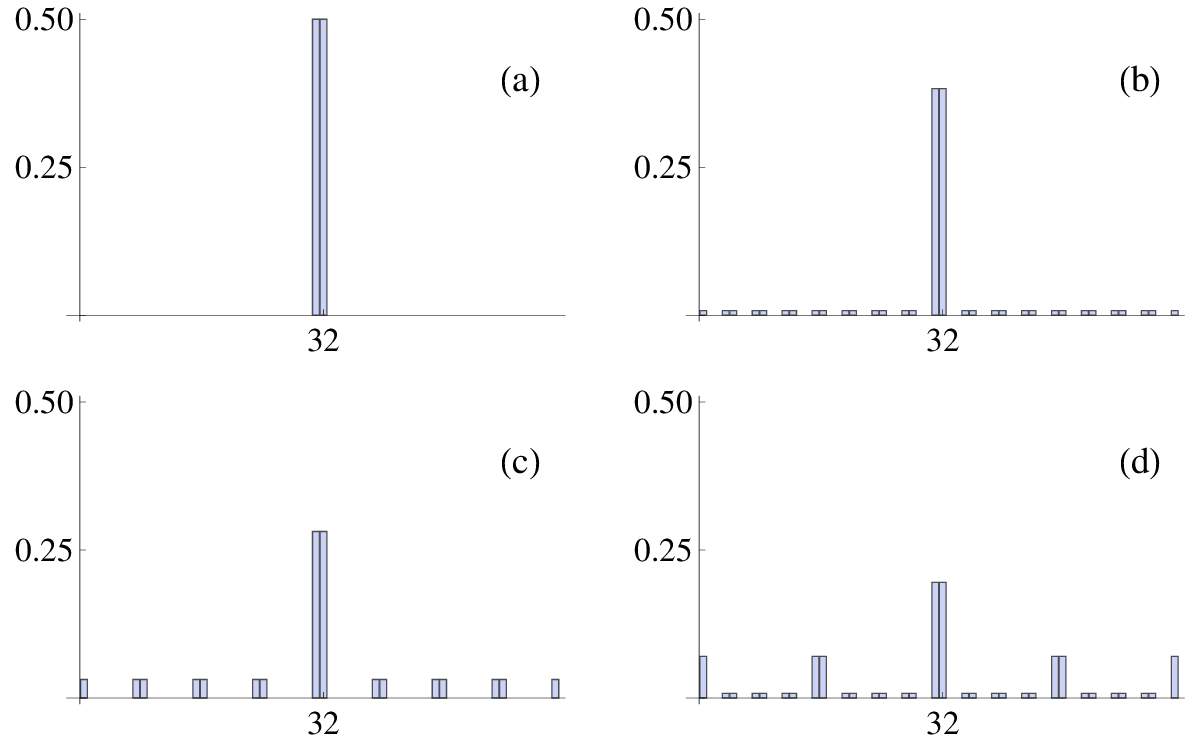}
            \end{center}
            \caption{For $N=64$, the effect of $l$
            unrestricted errors, where none of the $l$ errors are of the restricted type, on the probability of
            the central components is demonstrated for
            (a) $l=0$, (b) $l=2$, (c) $l=4$, and (d) $l=6$. Fig.
            7(c) corresponds to the case that $l=N/16$. }\label{fig:comperrors}
\end{figure}

The probability of the two central components is
\begin{align}
\left|\alpha\right|^2\geq\frac{1}{2}\left(1-\frac{8l}{N}+\frac{16l^2}{N^2}\right).
\end{align} The equality holds for the worst case where none of the errors are of the restricted type.
If $l=N/16$, then
$\left|\alpha\right|^2\geq\frac{1}{2}\left(\frac{9}{16}\right)$.
The amplitudes of the two central states can be combined into a
single state with amplitude $\sqrt{2}\alpha$ by an appropriate
unitary operation. Since $l$ is less than $N/16$, the desired
state can be identified with probability~$2\left|\alpha\right|^2$,
which is at least~$\frac{9}{16}$.
\end{proof}

We use the results of Lemmas 1, 2 and 3 to show that
Claims~\ref{restrictedclaim} and~\ref{unrestrictedclaim} are true.\\
\\
\textbf{Proof of Claim~\ref{restrictedclaim}:}\\

By Lemmas~1 and~2,
$\left|\Psi_3\right\rangle=\left|\frac{1}{2}+j\right\rangle_s+
\left|-\frac{1}{2}-j\right\rangle_s$ for
$z\in\tilde{\Xi}^{(N)}_{j,\,\mathbb{Z}_{N/4}}$. Immediately prior
to the measurement step, we require a unitary operator
$\op{U}_{2\mapsto1}$ that maps this superposition of two states
into a single basis state such that
\begin{align}\label{eg:2to1Op}
\op{U}_{2\mapsto1}\left|\Psi_3\right\rangle=\left|\frac{1}{2}+j\right\rangle_s.
\end{align} Since we know that the unknown string $z$ is either in set $\tilde{\Xi}_{N/2-1}^{(N)}$ or in $\tilde{\Xi}_{k}^{(N)}$,
we wish to measure the outcome of the qudit $|s\rangle_s$ in the
spin basis. We define the projection operator~\cite{NC00}
\begin{align}\label{eq:Ms}
M_s=\left|s\right\rangle_s\left\langle s\right|,
\end{align}
and outcome probability is
\begin{align}
\Pr[s]=\left\langle\frac{1}{2}+j\right|_s
M_s\left|\frac{1}{2}+j\right\rangle_s.
\end{align}
If $j=N/2-1$, then $\Pr[s]=1$ and $z\in A$, and if $j\neq N/2-1$,
then $\Pr[s]=0$ and $z\in B$. Thus, the restricted close Hadamard
problem is solved with certainty in a single oracle query. \hfill
$\square$
\\
 \\
\textbf{Proof of Claim~\ref{unrestrictedclaim}:}\\

By Lemmas~1 and~2,
$\left|\Psi_3\right\rangle=\left|\frac{1}{2}+j\right\rangle_s+
\left|-\frac{1}{2}-j\right\rangle_s$ for
$z\in\tilde{\Xi}^{(N)}_{j,\,\mathbb{Z}_{N/4}}$. By Lemma~3, the
effect of including $l$ unrestricted bit errors on the output
state may be expressed as
\begin{align}
\left|\Psi_3\right\rangle=&\alpha\left(\left|\frac{1}{2}+j\right\rangle_s+
\left|-\frac{1}{2}-j\right\rangle_s\right)\nonumber\\
&+\beta\sum_{\kappa}\pm\left|\frac{1}{2}+k\right\rangle_s+\gamma\sum_{\lambda}\pm\left|\frac{1}{2}+k\right\rangle_s,
\end{align} where the symbols $\alpha$ and $\beta$ are given by Eqs.~(\ref{eq:compampprinciple})
and~(\ref{eq:compampnext}) respectively, and $|\gamma|<|\beta|$.
Note that $\kappa+\lambda=N-2$, so that all $N$ possible states
are accounted for, but the specific values of $\gamma$, $\kappa$
and $\lambda$ are dependent on $N$ and $l$. The outcome
probabilities of measuring the state $|s\rangle_s$ are
\begin{align}
\Pr[s]=\left\langle\frac{1}{2}+j\right|_s
M_s\left|\frac{1}{2}+j\right\rangle_s.
\end{align}
If $j=N/2-1$, then $\Pr[s]>\left|\alpha\right|^2$, and if $j\neq
N/2-1$, then $\Pr[s]<\left|\beta\right|^2$. Assuming that the
number of unrestricted errors $l$ is less than $N/16$, then an
error of $O(e^{-q})$ can be achieved by making $O(q)$ repetitions
of the algorithm~\cite{AHS09}. Thus, the unrestricted close
Hadamard problem is solved with arbitrarily small error
probability in a constant number of queries. \hfill $\square$\\
\\
\textbf{Proof of Theorem~\ref{CHPinSS}:}\\

In the restricted case, the algorithm output error probability is
only due to the error in the input distribution given by
Eq.~(\ref{Eq:CentralProbBound}). In the unrestricted case, the
algorithm output error probability is due to both the error in the
input distribution and the inability of the algorithm to
accurately distinguish output states when the codeword errors are
unrestricted.

For the restricted case, the total success probability is
\begin{align}
\Pr_{\checkmark}=\mathcal{P}_c^{(N)}\Pr[s]_{\text{\tiny{restricted}}}>0.96.
\end{align} Similarly for unrestricted case, the total success
probability is
\begin{align}
\Pr_{\checkmark}=\mathcal{P}_c^{(N)}\Pr[s]_{\text{\tiny{unrestricted}}}>0.54.
\end{align} Since for each of these cases $\Pr_{\checkmark}>0.5$, then an
error of $O(e^{-q})$ can be achieved by making $O(q)$ repetitions
of the algorithm~\cite{AHS09}. This completes the proof of
Theorem~\ref{CHPinSS}.\hfill $\square$

\subsection{Classical Algorithm}

In this subsection we compare the performance of any classical
algorithm to the performance of the quantum algorithm.

\begin{claim}\label{classicalclaim}Any classical
deterministic algorithm requires~$\Omega(n)$ oracle queries of the
bit positions to solve the close Hadamard problem with certainty,
\textnormal{even} if there are no bit errors. A randomized
algorithm with bounded error probability also requires~$\Omega(n)$
queries, \textnormal{even} if there are no bit errors.
\end{claim}

Claim~\ref{classicalclaim} follows from information theoretical
considerations. The goal of the classical strategy is to determine
which of the $N/2$ Hadamard codewords is loaded into the oracle.
The number of possible solutions is then initially $\Omega(2^n)$.
Whenever a classical strategy performs a query, it can eliminate
at most half of the remaining possible solutions, even if there
are no errors. To reduce the number of possible solutions to a
single solution, the classical strategy therefore requires at
least $\Omega(n)$ queries\footnote{See for example paragraph 6.1
in~\cite{Ata98} for an introduction to information theoretic lower
bounds.}. The lower bound also holds when the strings loaded into
the oracle are Hadamard codewords with errors.

In the next section we discuss how changing the unitary operators
$\op{R}$ and $\op{R}^\dag$ in the algorithm shown in
Fig.~\ref{fig:SVDJCircuit} changes the oracle decision problem
that can be solved.

\section{Alternative Algorithm}

The continuously-parameterized, finite-dimensional Hilbert space
of the spin system inspired an efficient algorithm for the
solution of the close Hadamard problem. The group structure of the
Hadamard codewords is implicit in the use of Hadamard operators in
the quantum algorithm. We now show that this computation model can
inspire other algorithms.  Other unitary operators can be employed
in the quantum circuit shown in Fig.~\ref{fig:SVDJCircuit}. The
discrete Fourier transform \cite{NC00} is an obvious alternative.
We provide a sketch of how the Fourier transform changes the group
structure of the codewords and point to the need for further
exploration of problems that could benefit from this computational
model.

We replace the  operators $\op{R}$ and $\op{R}^\dag$ in
Fig.~\ref{fig:SVDJCircuit} with the discrete Fourier transform
$\op{F}$ and $\op{F}^\dag$. The matrix representation of the
discrete Fourier transform is expressed as
\begin{equation}\label{matrixF-N}
\op{F}^{(N)}=\frac{1}{\sqrt{N}} \left(\begin{array}{cccc}
1         & 1                 &\cdots           &1       \\
1         & \omega            &\cdots        &\omega^{N-1}\\
\vdots    &\vdots             &\ddots         &\vdots     \\
1         &\omega^{(N-1)}    &\cdots     &\omega^{(N-1)(N-1)}\\
\end{array} \right),
\end{equation}
where $\omega=e^{\frac{\rm{i}\pi}{N}}$~\cite{NC00}.

In our analysis of the $\op{F}$-based algorithm, we adopt a
similar approach to that taken for the $\op{H}$-based algorithm
and define the `Fourier codewords' as
\begin{align}
\op{T}^{(N)}=\log_{(-1)}\left[\sqrt{N}\op{F}^{(N)}\right].
\label{Sec1EqnLogF}
\end{align}
Similar to the Hadamard codewords, the $j^{\rm th}$ Fourier
codeword is the $j^{\rm th}$ row of the matrix $\op{T}^{(N)}$. As
an example, we express the $N=8$ matrix as
\begin{align}
\op{T}^{(8)}=\left(
\begin{array}{cccccccc}
 0 & 0 & 0 & 0 & 0 & 0 & 0 & 0 \\
 0 & \frac{1}{4} & \frac{1}{2} & \frac{3}{4} & 1 & -\frac{3}{4} & -\frac{1}{2} & -\frac{1}{4}
   \\
 0 & \frac{1}{2} & 1 & -\frac{1}{2} & 0 & \frac{1}{2} & 1 & -\frac{1}{2} \\
 0 & \frac{3}{4} & -\frac{1}{2} & \frac{1}{4} & 1 & -\frac{1}{4} & \frac{1}{2} & -\frac{3}{4}
   \\
 0 & 1 & 0 & 1 & 0 & 1 & 0 & 1 \\
 0 & -\frac{3}{4} & \frac{1}{2} & -\frac{1}{4} & 1 & \frac{1}{4} & -\frac{1}{2} & \frac{3}{4}
   \\
 0 & -\frac{1}{2} & 1 & \frac{1}{2} & 0 & -\frac{1}{2} & 1 & \frac{1}{2} \\
 0 & -\frac{1}{4} & -\frac{1}{2} & -\frac{3}{4} & 1 & \frac{3}{4} & \frac{1}{2} & \frac{1}{4}
\end{array}
\right),\label{eqn:FourierCodeT8}
\end{align} with $T^{(8)}_4=01010101$. We see that the Fourier codewords are not bit strings
but rather can be thought of as fractional bits. These fractional
bits can still be encoded into the oracle function $\op{U}_z$
given in Eq.~(\ref{Eq:Oracle_Uf}).

We define what we term the \emph{simple Fourier codeword} oracle
decision problem and show that it can be solved in a single query
using the modified algorithm. Note that we have structured this
problem along the same lines as the close Hadamard problem with no
errors.
\begin{problem}\sl\label{FMAN}
Given the string $\check{A}=T^{(N)}_{N/2-1}$ and a set of strings
$\check{B}=T^{(N)}_{k}$,  with $k \in \left\{\mathbb{Z}_{N}\mid
k\neq N/2-1\right\}$ and a string $z$ randomly selected with
uniform distribution~$\mu$ such that $z\in_\mu \check{C}=\check{A}
\cup \check{B}$, the \textbf{Simple Fourier Codeword} problem is
to determine whether $z\in \check{A}$ or $z\in \check{B}$ with the
fewest oracle queries.
\end{problem}

The action of the algorithm on the input state is expressed as
\begin{align}
\left|\Psi_{3}\right\rangle=&{\op{F}^{\dag}}^{(N)}\op{U}_z
\op{F}^{(N)}\left(\left|\frac{1}{2}\right\rangle_s+
\left|-\frac{1}{2}\right\rangle_s\right).\label{eq:FAction}
\end{align} For $z=T^{(N)}_j$ and $j\in 0,1,\ldots,N-1$, the output state can be shown to be
\begin{align}
\left|\Psi_{3}\right\rangle=&\left|\frac{1}{2}+j\right\rangle_s+
\left|-\frac{1}{2}+j\right\rangle_s,\label{eq:FActionResult}
\end{align} where $\frac{1}{2}+j$ and $-\frac{1}{2}+j$ are modulo
$s$ sums in the sense that
$\left|\frac{1}{2}+\frac{N}{2}\right\rangle=\left|-s\right\rangle$.

We apply $U_{2\mapsto1}$ given in Eq.~(\ref{eg:2to1Op}) to the
state $\left|\Psi_{3}\right\rangle$ given in
Eq.~(\ref{eq:FActionResult}). Measuring the qudit $|s\rangle_s$ in
the spin basis with the measurement operator $M_s$ given in
Eq.~(\ref{eq:Ms}), distinguishes whether the encoded string is in
set $\check{A}$ or set $\check{B}$ thereby solving the simple
Fourier codeword problem in a single query.

The result given by Eq.~(\ref{eq:FActionResult}) is achieved by
exploiting group properties similar to those expressed in
Eqs.~(\ref{Sec2SubSA2WgrpIdent})
and~(\ref{eq:HadamardCodeWordSumProperty}) for the Hadamard
codewords. The columns of $\op{F}^{(N)}$ represent the
\emph{multiplicative} cyclic group of order $N$, where the
generator is the first non-trivial column of $\op{F}^{(N)}$. The
matrix of codewords $\op{T}^{(N)}$ represents the \emph{additive}
cyclic group of order $N$ as a result of taking the logarithm
of~$\op{F}^{(N)}$.

Each element of the group has the inverse relation
\begin{align}
T_j + T_{N-j}=T_0,\label{eq:TgrpIdent}
\end{align}
and each codeword also obeys the sum relation
\begin{align}
T_{j} + T_{N/2-j}=T_{N/2},\label{eq:TCodeWordSumProperty}
\end{align} where $N-j$ and $N/2-j$ are understood to be modulo
$N$ sums. In Fig.~\ref{fig:Fprobdist} we clearly see that, unlike,
the Hadamard codewords that preserve the superposition of two
symmetric states, the Fourier codewords preserve the superposition
of two adjacent states.
\begin{figure}[tbp]
            \begin{center}
            \includegraphics[width=9 cm]{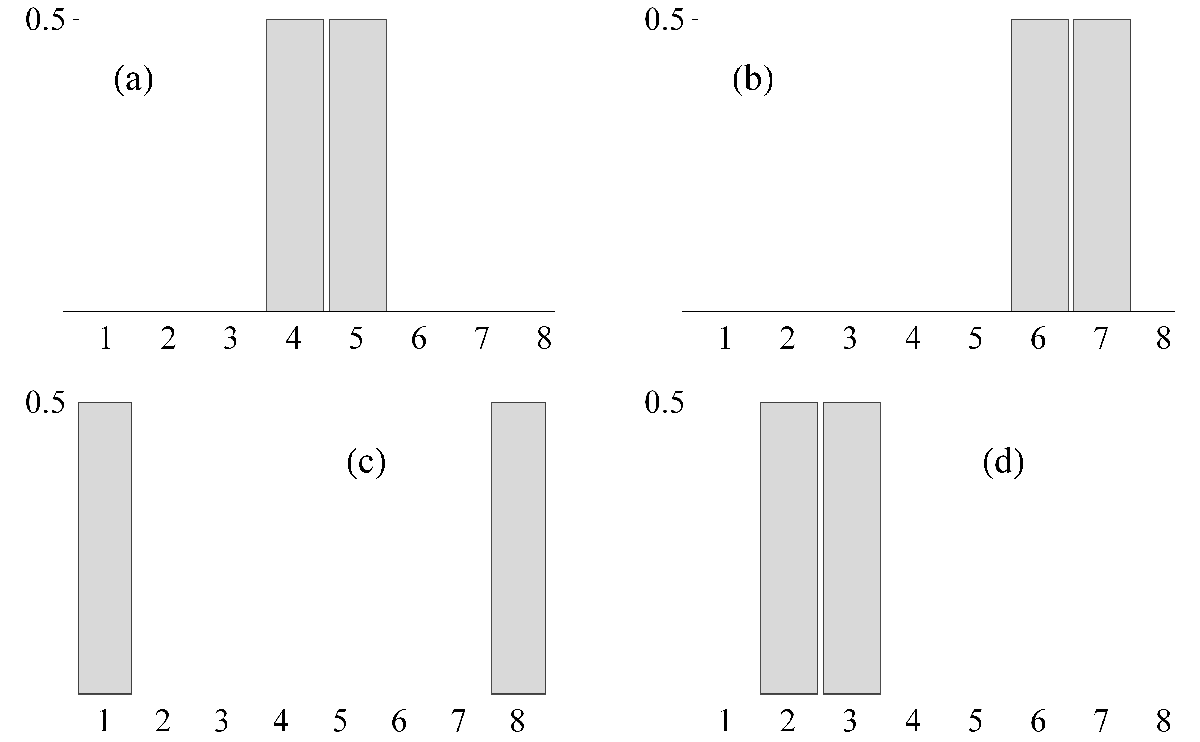}
            \end{center}
            \caption{For $s=7/2$, the Fourier-based algorithm probability distributions for the states given by
            Eq.~(\ref{eq:FActionResult}) for
            (a) $j=0$, (b) $j=2$ (c) $j=4$ and (d) $j=6$. The
            $\op{F}$-based algorithm preserves the `adjacency' of the input
            superposition, whereas the $\op{H}$-based algorithm preserves
            the `mirror' symmetry of the input superposition as
            shown in Fig.~\ref{Fig_for_Lemma1}.
            } \label{fig:Fprobdist}
\end{figure}

Comparison of the effect of the different operators is
interesting. The Hadamard codewords preserve symmetric
superpositions, and there are $N/2$ unique symmetric
superpositions. The Fourier codewords preserve adjacent
superpositions, and there are $N$ unique adjacent superpositions.

The exploration of the structure of error cancellations along the
lines of Eq.~(\ref{Sec2LmA22}) using the relationship
 Eq.~(\ref{eq:TCodeWordSumProperty}) for the Fourier codewords may lead to
new problems that can be efficiently solved using this model of
computation. For example, it is natural to apply this error
cancellation concept to the simple Fourier codeword oracle
decision problem presented in Problem \ref{FMAN} so that the
equivalent of Fourier codewords with errors may be included.

\section{Conclusions}

We have presented a model of continuous variable quantum
computation based on the continuously-parameterized yet
finite-dimensional Hilbert space of a spin system. Like continuous
variable quantum computation using the states of the harmonic
oscillator, this spin system is amenable to physical preparation
with linear rotation and quadratic squeezing operators. Unlike the
harmonic oscillator case where an arbitrary degree of squeezing
may be achieved, spin squeezing occurs in a spherical phase space,
and the minimum attainable variance of a spin component
asymptotically approaches the Heisenburg limit of
one-half~\cite{KU93}.

We have proposed a pilot oracle decision problem called the close
Hadamard problem to demonstrate the type of problem that is
amenable to solution using this computational model. We have used
a superposition of two discrete states to approximate the
optimally squeezed coherent spin state as input into a quantum
algorithm presented in Figure~\ref{fig:SVDJCircuit}, which is
adapted from the continuously-parameterized infinite-dimensional
single mode algorithm~\cite{AHS09,AHS12}.

We have shown that this model of computation solves the close
Hadamard problem with arbitrarily small error probability in a
constant number of oracle queries. The combination of the model
and the algorithm taken together describe a new model of quantum
computation.  Furthermore, the two versions of the close Hadamard
problem hi-light an interesting feature of the model.

The tolerance of errors observed in the restricted case is due to
error cancellation, which results from employing the symmetric
superposition of spin states as algorithm input combined with the
group structure of the Hadamard codewords and the employment of
Hadamard operators in the algorithm. This relationship between the
structure of the information associated with the problem and the
structure of the unitary operators employed in the quantum
algorithm is what enables us to solve the restricted problem more
efficiently in this computational model. It may prove fruitful to
explore the relationship between the error cancellation observed
here and how the perturbations related to Grover's problem
mutually cancel in arbitrarily high-dimensional search
spaces~\cite{Jin15}.

We have further explored the relationship between operators and
codewords and have shown that it can inspire other algorithms. We
have shown that using the discrete Fourier transform as an
alternative to the Hadamard operator changes the group structure
of the codewords.  This in turns indicates that further
exploration may lead to the discovery of new problems that would
benefit from this computational model. We conclude that the
continuously-parameterized representation of quantum dynamical
systems having a finite-dimensional Hilbert space gives us a new
model of quantum computation that is worthy
of further exploration.\\
\section*{Acknowledgements}

We appreciate financial support from the Alberta Ingenuity Fund
(AIF), Alberta Innovates Technology Futures (AITF), Canada's
Natural Sciences and Engineering Research Council (NSERC), the
Canadian Network Centres of Excellence for Mathematics of
Information Technology and Complex Systems (MITACS), and the
Canadian Institute for Advanced Research (CIFAR).

\bibliographystyle{aiaa-doi}

\end{document}